\documentclass{llncs}

\usepackage{amsmath}
\usepackage{tikz}
\usepackage{expdlist}
\usepackage{paralist}
\usepackage{algorithm,algorithmicx,algpseudocode, enumerate}

\DeclareMathOperator{\lcp}{lcp}
\DeclareMathOperator{\lcs}{lcs}
\newcommand{\suf}{\mathit{Suf}}
\renewcommand{\epsilon}{\varepsilon} 

\newcommand{\Oh}{\mathcal{O}}
\newcommand{\floor}[1]{\left\lfloor #1 \right\rfloor}

\newtheorem{fact}[theorem]{Fact}
\newtheorem{observation}[theorem]{Observation}
\newcommand{\eps}{\varepsilon}
\newcommand{\nlarger}{\mathit{NotLarger}}

\begin{document}
\title{Substring Suffix Selection}
\titlerunning{Substring Suffix Selection}

\author{
Maxim Babenko\inst{1}
\and
Pawe{\l} Gawrychowski\inst{2}
\and
Tomasz Kociumaka\inst{3}
\and
Tatiana Starikovskaya\inst{1}
}

\institute{
Higher School of Economics, Moscow, Russia, \email{maxim.babenko@gmail.com,tat.starikovskaya@gmail.com}
\and
Max-Planck-Institut f\"{u}r Informatik, Saarbr\"{u}cken, Germany, \email{gawry@cs.uni.wroc.pl}
\and
Institute of Informatics, University of Warsaw, Warsaw, Poland, \email{kociumaka@mimuw.edu.pl}
}






\date{\empty}
\maketitle

\begin{abstract}
We study the following \emph{substring suffix selection} problem: given a substring of a string $T$ of length $n$, compute its $k$-th lexicographically smallest suffix. This a natural generalization of the well-known question of computing the maximal suffix of a string,
which is a basic ingredient in many other problems.

We first revisit two special cases of the problem, introduced by Babenko, Kolesnichenko and Starikovskaya~[CPM'13], in which we are asked to compute the minimal non-empty and the maximal suffixes of a substring. For the maximal suffixes problem, we give a linear-space structure with
$O(1)$ query time and linear preprocessing time, i.e., we manage to achieve optimal construction and optimal query time simultaneously.
For the minimal suffix problem, we give a linear-space data structure with $O(\tau)$ query time and $O(n \log n / \tau)$ preprocessing time, where $1 \le \tau \le \log n$ is a~parameter of the data structure. As a
sample application, we show that this data structure can be used to compute the Lyndon decomposition of any substring of $T$ in $O(k \tau)$ time, where $k$ is the number of distinct factors in the decomposition.

Finally, we move to the general case of the substring suffix selection problem, where using any combinatorial properties seems more difficult.
Nevertheless, we develop a linear-space data structure with $O(\log^{2+\varepsilon} n)$ query time.
\end{abstract}

\section{Introduction}
    \label{sec:intro}
    Computing the $k$-th lexicographically smallest suffix of a~string is both an interesting problem on its own, and a crucial ingredient
in solutions to many other problems. As an example of the former, a well-known result by Duval~\cite{Duval} is that the maximal suffix of a~string can be found in linear time
and constant additional space. As an example of the latter, the famous constant space pattern matching algorithm of Crochemore-Perrin
is based on the so-called critical factorizations, which can be found by looking at maximal suffixes~\cite{AlgorithmsOnStrings}.
 In the more general version, a straightforward way to compute the $k$-th suffix of a~string is to
construct its suffix array, which results in a linear time and space solution, assuming that we can sort the letters in linear time.
Surprisingly, one can achieve linear time complexity even without such assumption, as shown by Franceschini and Muthukrishnan~\cite{SufSelect}.

We consider a natural generalization of the question of locating the $k$-th suffix of a~string. We assume that the string we are asked to compute the
$k$-th suffix for is actually a substring of a longer text $T[1..n]$ given in advance. Information about $T$ can be preprocessed and then used 
to significantly speed up the computation of the desired suffixes of a query string. This seems to be a very natural setting whenever we are thinking
about storing large collections of text data. Other problems studied in such version include the substring-restricted pattern matching, where we are 
asked to return occurrences of a given word in some specified interval~\cite{SRR}, the factor periodicity problem, where we are asked to compute the 
period of a given substring~\cite{FactorPeriodicity2012}, and substring compression, where the goal is to output compressed representation of
given substring~\cite{SubstringCompress,GenSubstringCompress}.

We start with two special cases of the problem, namely, computing the minimal non-empty and the maximal suffixes of a~substring of $T$. These two 
problems were introduced in~\cite{Minmaxsuf}. The authors proposed two linear-space data structures for~$T$. Using the first data structure, one can 
compute the minimal suffix of any substring of $T$ in $\Oh(\log^{1+\eps} n)$ time. The second data structure allows to compute the maximal suffix of a 
substring of $T$ in $\Oh(\log n)$ time. Here we improve upon both of these results. First, we describe a series of linear-space data structures that allow, 
for any $1 \le \tau \le \log n$, to compute the minimal suffix of a substring of $T$ in $\Oh(\tau)$ time. Construction time is $\Oh(n \log n / \tau)$. 
Secondly, we describe a linear-space data structure for the maximal suffix problem with $\Oh(1)$ query time. The data structure can be constructed in linear time. Computing the maximal or the minimal suffix is a fundamental tool used in more complex algorithms, so our results can hopefully be used
to efficiently solve also other problems in such setting, i.e., when we are working with substrings of some long text $T$. As a~particular application,
we show how to compute the Lyndon decomposition~\cite{chen1958free} of a~substring of $T$ in $\Oh(k \tau)$ time, where $k$ is the number of
distinct factors in the decomposition. 

We then proceed to the general case of the problem, which is much more interesting from the practical point of view. It is also substantially
more difficult, mostly because the $k$-th suffix of a substring does not enjoy the combinatorial properties the minimal and the maximal suffixes have. Nevertheless, we are able to propose a linear-space data structure with $\Oh(\log^{2+\varepsilon} n)$ query time for the general case.

Our data structures are designed for the standard word-RAM model, see~\cite{AHU-74} for a definition.
We assume that letters in $T$ can be sorted in $O(n)$ time.

\section{Preliminaries}
   \label{sec:prelim}
   We start by introducing some standard notation and definitions. 
Let $\Sigma$ be a finite ordered non-empty set (called an \emph{alphabet}).
The elements of~$\Sigma$ are \emph{letters}.

A finite ordered sequence of letters (possibly empty) is called a \emph{string}. 
Letters in a~string are numbered starting from~1, that is, a string $T$ of \emph{length} $k$ consists of letters $T[1], T[2], \ldots, T[k]$. The length~$k$ of $T$ is denoted by~$|T|$. 
For $i \le j$, $T[i..j]$ denotes the \emph{substring} of $T$ from position $i$ to position $j$ (inclusive). 
If $i > j$, $T[i..j]$ is defined to be the empty string. 
Also, if $i = 1$ or $j = |T|$ then we omit these indices and we write just $T[..j]$ and $T[i..]$. 
Substring $T[..j]$ is called a \emph{prefix} of $T$, and $T[i..]$ is called a \emph{suffix} of $T$. 
A \emph{border} of a string~$T$ is a string that is both a prefix and a suffix of~$T$ but differs from $T$. 

A string $T$ is called \emph{periodic with period $\rho$} if $T = \rho^s \rho'$ for an integer $s \ge 1$ and a (possibly empty) proper prefix $\rho'$ of $\rho$. 
When this leads to no confusion, the length of $\rho$ will also be called a period of $T$.
Borders and periods are dual notions; namely, if $T$ has period~$\rho$ then it has a border of length $|T| - |\rho|$, and vice versa (see, e.g.,~\cite{TextAlgorithms}).

Letters are treated as integers in a range $\{1, \ldots, |\Sigma|\}$; a pair of letters can be compared in $\Oh(1)$ time. 
This \emph{lexicographic} order over $\Sigma$ is linear and can be extended in a standard way to the set of strings in $\Sigma$. 
Namely, $T_1 \prec T_2$ if either (i) $T_1$ is a prefix of $T_2$; or (ii) there exists $0 \le i < \min(|T_1|, |T_2|)$ such that $T_1[..i] = T_2[..i]$, and $T_1[i+1] \prec T_2[i+1]$.

\section{Suffix Array and Related Data Structures}
Consider a fixed string $T$.
For $1\le i<j\le |T|$ let $\suf[i,j]$ denote $\{T[i..],\ldots,T[j..]\}$.
The set $\suf[1,|T|]$ of all non-empty suffixes of $T$ is also denoted as $\suf$.
The \emph{suffix array} $SA$ of a string~$T$ is a permutation of $\{1, \ldots, |T|\}$ defining the lexicographic order on $\suf$. 
More precisely, $SA[r] = i$ if the rank of $T[i..]$ in the lexicographic order on $\suf$ is $r$. 
The inverse permutation is denoted by $ISA$; it reduces lexicographic comparison of suffixes $T[i..]$ and $T[j..]$ to integer comparison of their ranks $ISA[i]$ and $ISA[j]$. 
For a string $T$, both $SA$ and $ISA$ occupy linear space and can be constructed in linear time (see~\cite{PST-07} for a survey). 
For strings $S,T$ we denote the length of their longest common prefix by $\lcp(S,T)$,
and of their longest common suffix by $\lcs(S,T)$.

While $SA$ and its reverse are useful themselves, equipped with additional data structures they are even more powerful.
We use several classic applications listed below.
\begin{lemma}\label{lem:all}
A string $T$ of length $n$ can be preprocessed in $\Oh(n)$ time so that the following queries can be answered in $\Oh(1)$ time:
\begin{enumerate}[(a)]
  \item\label{it:lcp} given substrings $x$, $y$ compute $\lcp(x,y)$ and determine if $x\prec y$,
  \item\label{it:suf} given indices $i,j$ compute the \emph{maximal} and \emph{minimal} suffix
  in $\suf[i,j]$,
\end{enumerate}

\end{lemma}
\begin{proof}
Queries~(\ref{it:lcp}) is a classic application of the LCP array equipped with the data structure for range minimum queries,
see~\cite{AlgorithmsOnStrings} for details.
Queries (\ref{it:suf}) are just range minimum (maximum) queries on $ISA$,
it suffices to equip $ISA$ with the appropriate data structure~\cite{Bender:2000:LPR:646388.690192}.
\end{proof}
These simple queries can be used to answer more involved ones.
\begin{lemma}\label{lem:pow}
The enhanced suffix array can answer the following queries in constant time. Given substrings $x,y$ of~$T$
compute the largest integer $\alpha$ such that $x^\alpha$ is a prefix of $y$.
\end{lemma}
\begin{proof}
Let $y=T[i..j]$. If $\lcp(x,y)<|x|$, then the answer is clearly 0.
Otherwise, we claim  $\alpha = 1+\floor{\frac{\lcp(T[i..j], T[i+|x|..j])}{|x|}}$.
Indeed, if $y=x^\alpha z$, then $T[i+|x|..j]= x^{\alpha-1} z$, i.e.
$\lcp(T[i..j], T[i+|x|..j]) = |x|(\alpha-1)+\lcp(xz,z)<|x|\alpha$, since $\lcp(x,z)<|x|$ by maximality of $\alpha$.
On the other hand a simple inductive argument shows that $\lcp(xT[i+|x|..j], T[i+|x|..j])\ge k|x|$
implies that $x^{k+1}$ is a prefix of $T[i..j]=xT[i+|x|..j]$.
\end{proof}
Note that the queries on the enhanced suffix array of $T^R$, the reverse of $T$,
are also meaningful in terms of $T$.
In particular for a pair of substrings $x,y$ we can compute $\lcs(x,y)$
and the largest  integer $\alpha$ such that $x^\alpha$ is a suffix of $y$.

\section{Minimal Suffix}
	\label{sec:min_suffix}
	Consider a string $T$ of length $n$. In this section we first describe a linear-space data structure for $T$ that can be constructed in $\Oh(n \log n)$ time and allows to compute the minimal non-empty suffix of any substring $T[i..j]$ in $\Oh(1)$ time. Then we explain how to modify the data structure to obtain $\Oh(n \log n / \tau)$ construction time and $\Oh(\tau)$ query time for any $1 \le \tau \le \log n$.

For each $j$ we select $\Oh(\log n)$ substrings $T[k..j]$, which we call \emph{canonical}. We denote the $\ell$-th longest canonical substring ending at position $j$ by $S^\ell_j$. The substring  $S^1_j$ is $T[j..j] = T[j]$. For $\ell \ge 2$ we set $m = \floor{\ell/2} - 1$ and define $S_j^\ell$ so that

$$\left|S_j^\ell\right| = 
\begin{cases}
2 \cdot 2^{m} + ( j \bmod{2^{m}}) \text{ if } \ell \text{ is even},\\
3 \cdot 2^{m} + ( j \bmod{2^{m}}) \text{ otherwise}.
\end{cases}
$$

Note that the number of such substrings is logarithmic for each $j$.
Moreover, if we split $T$ into chunks of size $2^m$ each,
then $S_j^\ell$ will start at the boundary of one of these chunks.
This alignment property will be crucial for the construction algorithm. 
Below we explain how to use canonical substrings to compute the minimal suffix of $T[i..j]$.
We start with two auxiliary facts.

\begin{fact}\label{fct:ratio}
For any $S_{j}^\ell$ and $S_{j}^{\ell+1}$ with $\ell\ge 1$ we have $\big|S_{j}^{\ell+1}\big| < 2\big|S_{j}^{\ell}\big|$.
\end{fact}
\begin{proof}
For $\ell=1$ the statement holds trivially. Consider $\ell \ge 2$. Let $m$, as before, denote $\floor{\ell/2} - 1$. If $\ell$ is even, then $\ell+1$ is odd and we have
$$ \left|S_{j}^{\ell+1}\right| = 3\cdot 2^m + (j\bmod 2^m) < 4\cdot  2^m \le 2\cdot \left (2\cdot 2^m+(j\bmod 2^m)\right) = 2\left|S_{j}^{\ell}\right|$$
while for odd $\ell$
\begin{center}$ \left|S_{j}^{\ell+1}\right|=2\cdot 2^{m+1} + (j\bmod 2^{m+1})  < 3\cdot 2^{m+1} \le 2\cdot \left (3\cdot 2^m+(j\bmod 2^m)\right) = 2\left|S_{j}^{\ell}\right|.$\end{center}
\vspace{-.7cm}
\end{proof}

For a pair of integers $1\le i < j \le n$, define $\alpha(i,j)$ to be the largest
integer $\ell$ such that $S_{j}^\ell$ is a proper suffix of $T[i..j]$. 

\begin{fact}\label{fct:comp}
Given integers $1\le i< j \le n$, the value $\alpha(i,j)$ can be computed in constant time.
\end{fact}
\begin{proof}
Let $m=\floor{\log|T[i..j]|}$.
Observe that $$\left|S_{j}^{2m-1}\right|=3\cdot 2^{m-2} + (j\bmod 2^{m-2})<2^{m}\le |T[i..j]|$$
and $$\left|S_{j}^{2m+2}\right| = 2\cdot 2^{m} + (j \bmod 2^{m}) \ge 2^{m+1} > |T[i..j]|.$$
Consequently $\alpha(i,j)$ is equal to $2m-1$, $2m$ or $2m+1$ which can be verified in constant time.
\end{proof}

\begin{lemma}
\label{lm:minsuffix}
	The minimal suffix of $T[i..j]$ is either equal to
	\begin{enumerate}[(a)]
      \item $T[p..j]$, where $p$ is the starting position of the minimal suffix in $\suf [i, j] = \{T[i..], T[i+1..], \ldots, T[j..]\}$, or
	  \item the minimal suffix of $S_{j}^{\alpha(i,j)}$.
\end{enumerate}
\end{lemma}
\begin{proof}
	By Lemma~$1$ in~\cite{Minmaxsuf} the minimal suffix is either equal to $T[p..j]$ or to its shortest non-empty border. Moreover, in the latter case the length of the minimal suffix is at most $\frac{1}{2} |T[p..j]| \le \frac{1}{2} |T[i..j]|$. On the other hand, by Fact~\ref{fct:ratio} the length of $S_{j}^{\alpha(i,j)}$ is at least $\frac{1}{2} |T[i..j]|$. Thus, in the second case the minimal suffix of $T[i..j]$ is the minimal suffix of $S_{j}^{\alpha(i,j)}$.
\end{proof}

Recursively applying Lemma~\ref{lm:minsuffix} we obtain the following

\begin{corollary}\label{cor:min-main}
	For $\ell=1,\ldots,\alpha(i,j)$ let $p_\ell$ be the minimal suffix in $Suf[j-|S_{j}^\ell|+1,j]$, and let $p$ be the minimal suffix in $Suf[i,j]$.
	The minimal suffix of $T[i..j]$ starts at one of the positions in $\{p, p_0,p_1,\ldots,p_{\alpha(i,j)}\}$.
\end{corollary}

With some knowledge about the minimal suffixes of canonical substrings, the set of candidate positions can be reduced.

\begin{observation}\label{obs:min-main}
	For any $k,k'$ such that $1\le k'<k\le \alpha(i,j)$: 
	\begin{enumerate}[(a)]
  		\item if the minimal suffix of $S_{j}^{k}$ is longer than $\big|S_j^{k'}\big|$, then positions $p_0,p_1,\ldots,p_{k'}$ do not need to be considered
  		as candidates in Corollary~\ref{cor:min-main},
   		\item if the minimal suffix of $S_{j}^{k}$ is not longer than $\big|S_j^{k'}\big|$, then positions $p_{k'+1},p_{k'+2},\ldots,p_{k}$ do not need to be considered as candidates in Corollary~\ref{cor:min-main}.
	\end{enumerate}
\end{observation}

We now explain how this result is used to achieve the announced time and space bounds.

\subsection{Data Structure}
Apart from the enhanced suffix array, we store, for each $j = 1,\ldots,n$, a bit vector $B_j$ of length $\alpha(1,j)$. Here $B_j[\ell] = 1$ if and only if $\ell = 1$ or the minimal suffix of $S^\ell_j$ is longer than $|S^{\ell-1}_j|$. Since $\alpha(1,j)=\Oh(\log j)$, each vector $B_j$ can be stored 
in a constant number of machine words, which gives $\Oh(n)$ space in total.

\subsection{Query}
To compute the minimal suffix of $T[i..j]$, we determine $\alpha=\alpha(i,j)$ (see Fact~\ref{fct:comp}) and locate the highest set bit $B_j[h]$ such that $1 \le h \le \alpha$. 
As $B_j[h]=1$ and $B_j[h']=0$ for $h'\in \{h+1,\ldots,\alpha\}$ Observation~\ref{obs:min-main} implies that
the minimal suffix starts at either $p$ or $p_h$. $ISA[p]$ is the minimum in $ISA[i,j]$, and $ISA[p_h]$ is the minimum in $ISA[j + 1 - |S_j^h|,  j]$. Hence the enhanced suffix array can be used to compute $p$ and $p_h$ as well as find the smaller of $T[p..j]$ and $T[p_h..j]$, all in $\Oh(1)$ time.

\subsection{Construction}
It suffices to  explain how vectors $B_j$ are computed. At the beginning we set all bits $B_j[1]$ to $1$. For each $m = 1,\ldots,\lfloor\log n\rfloor$ we compute the minimal suffixes of $S^\ell_j$, where $\floor{\ell/2}-1 = m$ and $1 \le j \le n$. To do this, we first divide $T$ into \emph{chunks} of length $2^m$. Each substring $S^\ell_j$ starts at the beginning of one of these chunks and has length smaller than $4\cdot 2^m$. Hence $S^\ell_j$ is a prefix of a substring consisting of at most four consecutive chunks. Recall that a variant of Duval's algorithm (see~Algorithm 3.1 in~\cite{Duval}) takes linear time to compute the lengths of minimal suffixes of all prefixes of a given string. We run this algorithm for each four (or less at the end) consecutive chunks and thus obtain the minimal suffixes of the substrings $S^\ell_j$, where $\floor{\ell/2}-1 = m$ and $1 \le j \le n$, in $\Oh(n)$ time. The value of $B_j[\ell]$ can now be found directly by comparing the length of minimal suffix of $S^\ell_j$ with $|S^{\ell-1}_j|$. Note that the space usage is $\Oh(n)$. We proved

\begin{theorem}
 	A string $T$ of length $n$ can be stored in an $\Oh(n)$-space structure that enables to compute the minimal suffix of any substring of $T$ in $\Oh(1)$ time. This data structure can be constructed in $\Oh(n \log n)$ time.
\end{theorem}

To obtain a data structure with $\Oh(n \log n / \tau)$ construction and $\Oh(\tau)$ query time, we define the bit-vectors in a slightly different 
way. We set $B_j$ to be of size $\floor{\alpha(1,j)/\tau}$ with $B_j[\ell]=1$ if and only if $\ell=1$ or
the minimal suffix of $S_{j}^{\tau\ell}$ is longer than $|S_{j}^{\tau(\ell-1)}|$.
This way we need only $\Oh(\log n/ \tau)$ phases in the construction algorithm, so it takes $\Oh(n \log n / \tau )$ time.

Again, let $p_\ell$ denote the starting position of the minimal suffix in $\suf[j - |S_j^\ell|+1, j]$. To compute the minimal suffix of $T[i..j]$, we determine $\alpha=\alpha(i,j)$ and locate the highest set bit $B_j[h]$, $\floor{\alpha/\tau} \ge h \ge 1$. Then, by Observation~\ref{obs:min-main} $B_j[h]=1$ and $B_j[h']=0$ for $h'\in\{h+1,\ldots,\floor{\alpha(i,j)/\tau} \}$ implies  the minimal suffix starts at one of the positions
$\{p, p_{(h-1)\tau+1},\ldots,p_{h\tau}, p_{\tau \floor{\frac{\alpha}{\tau}}},\ldots, p_{\alpha(i,j)}\}$.
Each of these positions can be computed in constant time, each two of the suffixes can be compared in constant time as well. That is, the data structure allows to compute the minimal suffix of any substring in $\Oh(\tau)$ time. Summing up,

\begin{theorem}
 	For any $1 \le \tau \le \log n$, a string $T$ of length $n$ can be stored in an $\Oh(n)$-space data structure that enables to compute the minimal suffix of any substring of $T$ in $\Oh(\tau)$ time. This data structure can be constructed in $\Oh(n \log n / \tau)$ time.
\end{theorem}

\subsection{Applications}
As a corollary we obtain an efficient data structure for computing Lyndon decompositions of substrings of $T$. We recall the definitions first. A string $w$ is said to be a \emph{Lyndon word} if and only if it is strictly smaller than its proper cyclic rotations. For a nonempty string $x$ a decomposition $x = w_1^{\alpha_1} w_2^{\alpha_2} \ldots w_k^{\alpha_k}$ is called a \emph{Lyndon decomposition} if and only if $w_1 > w_2 > \ldots > w_k$ are Lyndon words, see~\cite{chen1958free}.

\begin{lemma}[\cite{Duval}]\label{lem:duv}
	If $x = w_1^{\alpha_1} w_2^{\alpha_2} \ldots w_k^{\alpha_k}$ is a Lyndon decomposition, then $w_k$ is the minimal suffix of $x$.
\end{lemma}
\begin{lemma}
	Let $x = uv^\alpha$, where $v$ is the minimal suffix of $x$ and $u$ does not end with $v$. Let $u = z_1^{\beta_1} z_2^{\beta_2} \ldots z_\ell^{\beta_\ell}$ be the Lyndon decomposition of $u$. Then $x = z_1^{\beta_1} z_2^{\beta_2} \ldots z_\ell^{\beta_\ell} v^\alpha$ is the Lyndon decomposition of $x$.
\end{lemma}
\begin{proof}
	Any word admits a unique Lyndon decomposition~\cite{chen1958free}. Let $x = w_1^{\alpha_1} w_2^{\alpha_2} \ldots w_k^{\alpha_k}$ be the Lyndon decomposition of $x$. From Lemma~\ref{lem:duv} we obtain that $w_k=v$, moreover $w_{k-1}>w_k$ is the minimal suffix of $w_1^{\alpha_1} w_2^{\alpha_2} \ldots w_{k-1}^{\alpha_{k-1}}$, so $\alpha_k \ge \alpha$. Clearly $\alpha_k \le \alpha$, which proves equality.
	From the definition it follows that $u = w_1^{\alpha_1} w_2^{\alpha_2} \ldots w_{k-1}^{\alpha_{k-1}}$ is the Lyndon decomposition of $u$ and hence it coincides with the decomposition $u = z_1^{\beta_1} z_2^{\beta_2} \ldots z_\ell^{\beta_\ell}$. The claim follows.
\end{proof}

\begin{corollary}
 	For any $1 \le \tau \le \log n$ a string $T$ of length $n$ can be stored in an $\Oh(n)$-space data structure that enables to compute the Lyndon decomposition of any substring of $T$ in $\Oh( k\tau)$ time, where $k$ is the number of distinct factors in the decomposition. This data structure can be constructed in $\Oh(n \log n / \tau)$ time.
\end{corollary}

\section{Maximal Suffix}
	\label{sec:max_suffix}
	We now turn to the maximal suffix problem. Our solution is based on the following notion.
For $1\le p\le j\le |T|$ we say that a position $p$ is \emph{$j$-active} if there is no position $p'\in\{p+1,\ldots,j\}$ such that $T[p'..j] \succ T[p..j]$. 
Equivalently, $p$ is $j$-active exactly when the suffix $T[p..j]$ is the maximal suffix of some substring of $T$ ending at $j$.
From the definition it follows that a starting position of the maximal suffix of $T[i..j]$ is the minimal $j$-active position in $[i,j]$. 

\begin{example}
If $T[1..8]=\texttt{dcccabab}$, the $8$-active positions are $1, 2,3, 4, 6, 8$.
\end{example}

We will not store $j$-active positions for each $j$ explicitly because there can be too many of them. Instead we will consider, for each $1 \le j \le n$, a partition of an interval $[1,j]$ into a number of disjoint subintervals. 
For this partition we will keep a bit vector where set bits correspond to the subintervals containing $j$-active positions. Computing the maximal suffix 
of $T[i..j]$ will consist of three steps:  first, we compute the subinterval $i$ belongs to, call it $[\ell',r']$, and, using the bit vector, the leftmost 
subinterval completely to the right and containing a $j$-active position, call it $[\ell,r]$. Then the minimal $j$-active position must lie in one of these 
two subintervals. More precisely, it either lies in $[i,r']$, or in $[\ell,r]$. We separately compute the maximal suffix of $T[i..j]$ starting in these two subintervals, and return the lexicographically larger one.

\subsection{Data Structure}
Our data structure for computing maximal suffixes of substrings of $T$ consists of two parts. Partitions and bit vectors will be used to locate the first subinterval to the right of $i$ that contains a $j$-active suffix, and data structures associated with suffix arrays of $T$ and for the reverse of $T$ will be used  to compute the minimal $j$-active position in this subinterval.

\textbf{Nice partitions and bit vectors:} Nice partitions are defined recursively. The nice partition of $[1,j]$ consists of disjoint subintervals $B_{1},B_{2},\ldots,B_{\ell}$ and satisfies the following properties:
\begin{enumerate}
\item $|B_{1}| \le |B_{2}| \le \ldots \le |B_{\ell}|$;
\item Length of each subinterval is a power of two;
\item Lengths of each two consecutive subintervals are the same, or differ by a factor of two;
\item There are no three subintervals of equal length.
\end{enumerate}

The nice partition of an interval $[1,1]$ consists of the interval $[1,1]$ itself. Given a nice partition of $[1,j]$ we can create a nice partition of $[1,j+1]$ by adding a new interval $[j+1,j+1]$. Then it might be the case that there are three intervals of length $1$. In such case we merge the two leftmost ones into one of length $2$ and repeat until there are at most 2 intervals of each length. 
The result is a nice partition of $[1,j]$ satisfying properties 1-4.

For each $j$ we store a bit vector of length $\Oh(\log n)$ indicating which subintervals of the partition contain $j$-active positions. 

We will also make use of two pre-computed tables. For each $w \in \{0,1\}^{\floor{(\log n )/ 3}}$ and for each $\ell=1,\ldots,\floor{(\log n)/3}$ we store the number of set bits in a prefix of $w$ of length $\ell$ and the position of $w$ storing the $\ell$-th set bit. This way we can answer any rank/select query on a bit vector of length $\Oh(\log n)$ by a constant number of table look-ups. 

The second table will be used for locating the subinterval of the partition of $[1,j]$ containing~$i$. The partition of $[1,j]$ is completely determined by specifying $k$ such that last subinterval $B_\ell$ is of length $2^k$, and one word of length $\Oh(\log n)$, where the $t$-th bit is set when there are two blocks of length $2^t$ in the partition. We store the answers for each $w \in \{0,1\}^{\floor{(\log n )/ 3}}$ and for each possible position not larger than $n^{1/3}$. Again, we are able to process a query with a constant number of table lookups.

\subsection{Query}
Suppose that we are asked to find the maximal suffix of a substring $T[i,j]$. Recall that we want to do this in three steps: first, locate the subinterval $i$ belongs to, call it $[\ell',r']$, then find the leftmost interval on its right containing a $j$-active suffix, call it $[\ell,r]$. Using the second table we
compute the subinterval of the partition of $[1,j]$ containing $i$, and then we can use rank/select queries to retrieve the leftmost subinterval
to the right containing $j$-active position. Overall, the first step takes constant time.

The second step is to compute the lexicographically maximal suffix of $T[i..j]$ assuming that it starts in $[i,r']$, and the third step is to compute the
lexicographically maximal suffix of $T[i..j]$ assuming that it starts in $[\ell,r]$. Both these steps are actually very similar: it is enough to show how to
find the lexicographically maximal suffix of $T[i..j]$ assuming that it starts in $[\ell,r]$, where $|T[\ell,j]|\leq 2|T[r,j]|$ (such assumption follows
from the definition of a nice partition, where in the worst possible case $|B_{i}|=2^{i}$). For this we need some
combinatorial properties of maximal suffixes which we prove below. Let $T[\mu..j]$ be the desired lexicographically maximal suffix of $T[i..j]$. 
The goal will be to show that knowing the length of $|T[\mu..j]|$ up to a factor of two is actually enough to compute $\mu$ in $\Oh(1)$ time.

\begin{lemma}[Lemma~$2$ in~\cite{Minmaxsuf}]
	\label{lm:next_prefix}
	Let $P_1 = T[p_1..j]$ be a prefix of $T[\mu..j]$ and let $P_2 = T[p_2..j]$, where $T[p_2..]$ is the maximal suffix in $\suf [\ell, p_1-1]$. If $P_1$ is not a prefix of $P_2$, then $\mu = p_{1}$. Otherwise, $P_2$ is also a prefix of $T[\mu..j]$ and moreover $|P_2| > |P_1|$.
\end{lemma}
	


Let $T[p_{1}..]$ be the maximal suffix in $\suf [\ell, r]$ and $T[p_{2}..]$ the maximal suffix in $\suf [\ell, p_{1}-1]$. Assume that $T[\mu..j]$ starts 
somewhere in $[\ell, r]$, so that $P_{1}=T[p_{1}..j]$ is a prefix of $T[\mu..j]$. Define $P_{2}=T[p_{2}..j]$ and assume that $P_{1}$ is a prefix of $P_{2}$ (if 
not, the above lemma immediately gives us $\mu$). We state two more lemmas which describe the properties of such suffixes $P_{1}$ and $P_{2}$
when the length of $\rho = T[p_2..p_1-1]$ is smaller than $|P_{1}|$ (i.e., when $|P_{2}| \le 2 |P_{1}|$). These lemmas are essentially Lemmas~$4$ and~$5$ of~\cite{Minmaxsuf}, but because we use different notation, we repeat their proofs here. 

\begin{lemma}
\label{lm:our_periodicity}
With the notation above, $\rho$ is the shortest period of $P_{2}$, i.e., $T[p_2..j] = \rho^s \rho'$ where $s \ge 1$ and $\rho'$ is a proper prefix of $\rho$, and $\rho$ is the shortest string for which such decomposition exists. Moreover, actually $s \ge 2$.
\end{lemma}
\begin{proof}
Since $P_1$ is a border of $P_2$, $\rho = T[p_2..p_1-1]$ is a period of $P_2$. It remains to prove that no shorter period is possible. So, consider
the shortest period $\gamma$, and assume that $|\gamma| < |\rho|$. Then $|\gamma| + |\rho| \le 2 |\rho| \le |T[p_2..j]|$, and by the periodicity 
lemma substring~$P_2$ has another period $\gcd(|\gamma|, |\rho| )$. Since $\gamma$ is the shortest period, $|\rho|$ must be a multiple of $|\gamma|$, i.e., $\rho = \gamma^k$ for some $k \ge 2$. 

Suppose that $T[p_1..] \prec \gamma T[p_1..]$. Then prepending both parts of the latter inequality by copies of $\gamma$ gives $\gamma^{\ell-1} T[p_1..] \prec \gamma^\ell T[p_1..]$ for any $1 \le \ell \le k$, so from the transitivity of $\prec$ we get that	
$T[p_1..] \prec \gamma^k T[p_1..]=T[p_{2}..]$, which contradicts the maximality of $T[p_{1}..]$ in $\suf [\ell, r]$.
Therefore $T[p_1..] \succ \gamma T[p_1..]$, and consequently $\gamma^{k-1} T[p_1..] \succ \gamma^k T[p_1..]$.  But $\gamma^{k-1} T[p_1..] = T[p_2 + |\gamma|..]$ and $\gamma^k T[p_1..] = T[p_2..]$, so $T[p_2 + |\gamma|..]$ is larger than $T[p_2..]$ and belongs to $\suf [i, p_1-1]$, which is a contradiction.

The final observation that $s\ge 2$ follows from the condition that $|P_{2}| \le 2 |P_{1}|$.
\end{proof}
\begin{figure}[ht]
\begin{center}
\newcommand*\vv[2]{
\draw (#1, .3) -- (#1, 0);
\draw (#1, -.25) node[anchor=mid] {$#2$};
}
\newcommand*\vvd[2]{
\draw[dotted] (#1, .3) -- (#1, 0);
\draw (#1, -.25) node[anchor=mid] {$#2$};
}
\begin{tikzpicture}[scale=1]
\useasboundingbox (1,-.2) rectangle (10,.3);
\foreach \y in {0,.3}{
\draw[dashed] (1,\y) -- (1.75, \y);
\draw (1.75,\y) -- (9.25, \y);
\draw[dashed] (9.25,\y) -- (10, \y);
}
\vv{2}{i};
\vv{9}{j};
\vv{4.95}{p_2};
\vv{5.55}{p_1};
\vv{3.75}{\mu};
\draw  (5.25, 0.15) node {\footnotesize{$\rho$}};
\vvd{3.25}{\ell};
\vvd{6}{r};
\clip(3.5,-1) rectangle (9.25,1);
\foreach \x in {3.15,3.75, ..., 10}{
\draw (\x, .3) sin (\x+.3, .45) cos (\x+.6, .3);
}
\end{tikzpicture}%
\end{center}
\caption{\label{fig}
A schematic illustration of Lemma~\ref{lm:max}.
}
\end{figure}
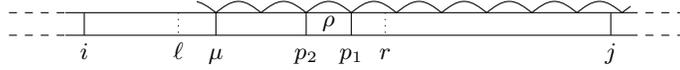

\begin{lemma}\label{lm:max}
Suppose that $P_2 = \rho P_1 = \rho^s \rho'$. 
If $|T[\mu..j]| \leq 2|P_1|$, then $T[\mu..j]$ is the longest suffix of $T[\ell..j]$
equal to $\rho^t\rho'$ for some integer $t$, see also Fig.~\ref{fig}
\end{lemma}
\begin{proof}
Clearly $P_2$ is a border of $T[\mu..j]$.
Since $P_2 = \rho P_1$ this implies $|T[\mu..j]| + |\rho| \leq 2|P_2|$.
Consequently the occurrences of $P_2$ as a prefix and as a suffix of $|T[\mu..j]|$ have an overlap with at 
least $|\rho|$ positions. As $|\rho|$ is a period of $P_2$, this implies that $|\rho|$ is also a period of $T[\mu..j]$. 
Thus $T[\mu..j]=\rho''\rho^r \rho'$, where $r$ is an integer and $\rho''$ is a proper suffix of $\rho$.
Moreover $\rho^2$ is a prefix of $T[\mu..j]$, since it is a prefix $P_2$, which is a prefix of $T[\mu..j]$.
Now $\rho''\ne \epsilon$  would imply a non-trivial occurrence of $\rho$ in $\rho^2$, which contradicts 
$\rho$ being primitive. 
Thus $T[\mu..j] = \rho^r\rho'$. If $t>r$, then $\rho^t\rho' \succ \rho^r \rho'$, so
$T[\mu..j]$ is the longest suffix of $T[i..j]$ equal to $\rho^t\rho'$ for some integer $t$.
\end{proof}

\begin{lemma}
Given a subinterval $[\ell, r]$ such that $|T[\ell,j]|\le 2|T[r..j]|$, and assuming that the lexicographically largest suffix
$T[\mu..j]$ of $T[i..j]$ starts there, we can compute $\mu$ in $\Oh(1)$ time.
\end{lemma}

\begin{proof}
Let $p_1$ be the starting position of the maximal suffix in $\suf [\ell, r]$, then $T[p_1..j]$ is a prefix of $T[\mu..j]$.
Let $p_2$ be the starting position of the maximal suffix in $\suf [\ell, p_1-1]$. $p_{1}$ and $p_{2}$ can be computed
$\Oh(1)$ time using two range maxima queries on $ISA$. Then we check if $T[p_{1}..j]$ is a prefix of $T[p_{2}..j]$.
If not, by Lemma~\ref{lm:next_prefix} $\mu=p_{1}$. Otherwise we set $\rho=T[p_2..p_1-1]$ and define $\nu$
to be  the largest integer such that $\rho^\nu\rho'$ is a suffix of $T[i..j]$ (or, equivalently, $\rho^\nu$ is a suffix of $T[i..p_1-1]$)
and set $p = p_{1} - \nu \cdot |\rho|$. Lemma~\ref{lem:pow} allows to compute $\nu$ in $\Oh(1)$ time using the enhanced suffix array of $T^{R}$.
The suffixes of $T[i..j]$ starting within $[\ell,r]$ are within multiplicative factor 2 from each other, so
Lemmas~\ref{lm:our_periodicity} and \ref{lm:max} imply $\mu = p$. 
\end{proof}

We apply the above lemma twice for the subintervals $[i,r']$ and $[\ell, r]$ found in the first step.
Finally, we compare the suffixes of $T[i..j]$ found in the second and third step in constant time, and return the larger one.
\subsection{Construction}
We start the construction with building the tables, which takes $o(n)$ time. In the main phase 
we scan positions of $T$ from the left to the right maintaining the list of active positions and computing the bit vectors.

We start with a lemma describing changes in the list of active suffixes upon a transition from $j$ to $j+1$.

\begin{lemma}
\label{lm:j-active-max}
	If the list of all $j$-active positions consists of $p_1 < p_2 < \ldots < p_k$, the list of $(j+1)$-active positions can be created by adding $j+1$, if $T[j+1] \succeq T[p_k]$ or $k=0$, and repeating the following procedure: if $p_\ell$ and $p_{\ell+1}$ are two neighbours on the current list, and $T[j+1] \neq T[j+1+p_{\ell}-p_{\ell+1}]$, remove $p_{\ell}$ or $p_{\ell+1}$ from the list, depending on whether $T[j+1] \succ T[j+1+p_{\ell}-p_{\ell+1}]$ or $T[j+1] \prec T[j+1+p_{\ell}-p_{\ell+1}]$, respectively.
\end{lemma}
\begin{proof}
	First we prove that if a position $1 \le p \le j$ is not $j$-active, then it is not $(j+1)$-active either. Indeed, if $p$ is not $j$-active, then by the definition there is a position $p < p' \le j$ such that $T[p..j] \prec T[p'..j]$. Consequently, $T[p..j+1] = T[p..j]T[j+1] \prec T[p'..j]T[j+1] = T[p'..j+1]$ and $p$ is not $(j+1)$-active. Hence, the only possible candidates for $(j+1)$-active positions are $j$-active positions and a position $j+1$. 

	Secondly, note that if $1 \le p \le j$ is a $j$-active position and $T[p'..j]$ is a prefix of $T[p..j]$, then $p'$ is $j$-active too. Suppose the contrary. Then there exists a position $p''$, $p' < p'' < j$, such that $T[p'..j] \prec T[p''..j]$, and it follows that $T[p..j] = T[p'..j] T[p+(j-p'+1)..j] \prec T[p''..j]$, a contradiction.

	A $j$-active position $p$ is not $(j+1)$-active only if (1) $T[j+1] \succeq T[p]$ or (2) there exists $p < p' \le j$ such that $T[p'..j]$ is a prefix of $T[p..j]$, i.e., $p'$ is $j$-active, and $T[p'..j+1] \succ T[p..j]$, or, equivalently, $T[j+1] \succ T[j+1+p-p']$. Both of these cases will be detected by the deletion procedure.
\end{proof}

\begin{example}
If $T[1..9]=\texttt{dcccababb}$,  the $8$-active positions are $1, 2, 3, 4, 6, 8$, and the $9$-active positions are $1, 2, 3, 4, 8, 9$, i.e., we add $9$ to the list of $8$-active positions, and then remove $6$.
\end{example}

The list of active positions will be maintained in the following way. 
After transition from the list of $j$-active positions to the list of $(j+1)$-active positions new pairs of neighbouring positions appear. For each such pair $p_\ell, p_{\ell+1}$ we compute $L=\lcp(T[p_\ell..], T[p_{\ell+1}..])$ and hence the smallest $j' = p_\ell + L$ when one of them should be removed from the list, and add a pointer from $j'$ to the pair $p_\ell, p_{\ell+1}$. 

When we actually reach $j=j'$, we check if $p_\ell$ and $p_{\ell+1}$ are still neighbours. If they are, we remove the appropriate element from the current list. Otherwise we do nothing. From Lemma~\ref{lm:j-active-max} it follows that the two possible updates of the list under transition from $j$ to $(j+1)$ are adding $(j+1)$ or deleting some position from the list. This guarantees that the process of deletion described in Lemma~\ref{lm:j-active-max} and the process we have just described are actually equivalent.

Suppose that we already know the list of $j$-active positions, the bit vector describing the nice partition of $[1,j]$, and the number of $j$-active positions in each subinterval of the partition. First we update the list of $j$-active positions. When a position is deleted from the list, we use the pre-computed table to find the subinterval the position belongs to, and decrement the counter of active positions in this subinterval. If the counter becomes equal to zero, we set the corresponding bit of the bit vector to zero. Then we start updating the partition: first we append a new subinterval $[j+1,j+1]$ to the partition of $[1..j]$ and initialize the counter of active positions in this subinterval by one. If then we have three intervals of length $1$, we merge the two leftmost ones into one interval of length $2$, add their counters, update the bit vectors, and repeat, if necessary. All these operations will take $\Oh(1)$ amortized time.

\begin{theorem}
A string $T$ of length $n$ can be stored in an $\Oh(n)$-space structure that allows computing the maximal suffix of any substring of $T$ in $\Oh(1)$ time. The data structure can be constructed in $\Oh(n)$ time.
\end{theorem}

\section{General Substring Suffix Selection}
	\label{sec:k-suffix}
	In the previous sections we considered the problems of computing the minimal and the maximal suffixes of a substring. 
Here we develop a data structure for the general case of the suffix selection problem.
Recall that the query, given a substring $T[i..j]$ and an integer $k$, returns the (length of) the $k$-th smallest suffix of $T[i..j]$.

For strings $S,T$
 we define $\nlarger(T,S)$ as the number of suffixes of $T$ not larger than $S$.
 Our data structure is based on the following fact.
\begin{fact}
Let $s_k$ be the $k$-th smallest suffix of $T[i..j]$ and let $S$ be the minimal suffix of $T$
such that $k'=\nlarger(T[i..j], S)\ge k$. Then $s_k$ is a prefix of $S$, and
there are exactly $k'-k$ longer prefixes of $S$ which are simultaneously suffixes of $T[i..j]$.
\end{fact}
\begin{proof}
Let $s_k=T[m_k..j]$.
Observe that $S\preceq T[m_k..]$ and $s_k \preceq S$, so
 $T[m_k..j] \preceq S\preceq T[m_k..]$ which means that $s_k=T[m_k..j]$ is indeed a prefix of $S$.
A similar reasoning shows that $s_{\ell}$ is a prefix of $S$ for each $\ell \in\{k+1,\ldots,k'\}$.
Conversely, any suffix of $T[i..j]$ larger than $s_k$ but not than $S$ must be a prefix of $S$,
so $s_{k+1},\ldots,s_{k'}$ are exactly the $k'-k$ longer prefixes of $S$ simultaneously being suffixes of $T[i..j]$.
\end{proof}

The query algorithm performs a binary search to determine $S$, calling a subroutine
to compute $\nlarger(T[i..j],S)$.
Then for $q=\nlarger(T[i..j],S)-k+1$  it finds the $q$-th largest prefix of $S$ simultaneously being a suffix
of $T[i..j]$. 

The second step is performed using Prefix-Suffix Queries, defined as follows.
For given substrings $S$, $S'$ of $T$ we are supposed to find (the lengths of) all prefixes
of $S$ which are simultaneously suffixes of $S'$. The lengths are reported as a sequence $A_1,\ldots,A_\ell$ of sets such 
that $\ell=\Oh(\log n)$, for each $i$ values in $A_i$ form an arithmetic progression and each element of $A_{i+1}$
is larger than each element of $A_i$.

\begin{lemma}\label{lem:prefsuf}
For any $\eps>0$ Prefix-Suffix Queries can be answered in $\Oh(\log^{1+\eps}n)$ time
by a data structure of size $\Oh(n)$.
\end{lemma}
\begin{proof}
In~\cite{FactorPeriodicity2012} Kociumaka et al. considered similar queries, where (the lengths of) all borders
of a given substring were reported. Here it suffices to store their data structure
for $T^2=TT$. Given a query with $S=T[i..j]$ and $S'=T[i'..j']$ it is enough find borders of $T[i..]T[..j']$
and filter out those longer than $\min(|S|,|S'|)$.
\end{proof}

Now it suffices to show how $\nlarger(T[i..j],S)$ can be efficiently computed.

\begin{lemma}
For any $\eps>0$ and a string $T$ of length $n$, there is an $\Oh(n)$-size data structure that given integers $i,j,\ell$ computes $\nlarger(T[i..j],T[\ell..])$
in $\Oh(\log^{1+\eps} n)$ time.
\end{lemma}
\begin{proof}
Note that if instead of counting suffixes of $T[i..j]$ which are not larger than $T[\ell..]$ we were to count suffixes 
in $\suf[i,j]$ which are not larger than $T[\ell..]$, our problem could be immediately
reduced to 2D orthogonal range counting on a set $Q=\{(m, ISA[m]):1\le m\le |T|\}$, the query rectangle would be $[i,j]\times [1,ISA[\ell]]$.
While our queries require more attention, we still use the data structure of~\cite{DBLP:conf/isaac/JaJaMS04},
which stores $Q$ using $\Oh(n)$ space and answers range counting queries in $\Oh(\frac{\log n}{\log \log n})$ time.

Observe that the number of suffixes of $T[i..j]$ smaller than $T[\ell..]$ is equal to the number of suffixes in $\suf[i, j]$ smaller than $T[\ell..]$ plus the number of suffixes in $\suf[i,j]$ which are bigger than $T[\ell..]$, but trimmed at the position $j$ become smaller than $T[\ell..]$ (i.e. trimmed suffixes become prefixes of $T[\ell..]$). The first term is determined using range counting as described above, while
computing the second one is a bit trickier. 
We use Prefix-Suffix Queries to find all suffixes of $T[i..j]$ which are simultaneously prefixes of $T[\ell..]$, and for each
arithmetic progression reported we count suffixes that are bigger than $T[\ell..]$.

Consider one of the progressions. 
Let $r < r+d < \ldots < r+\nu d$ be the starting positions as suffixes of $T[i..j]$. 
Then all substrings $T[r + (s-1)d..r+sd-1]$ are equal to $\rho = T[\ell..\ell+d-1]$. 
This means that $T[r+sd..]$, $s=0,1,\ldots, \nu$, can be represented as $\rho^{\nu'-s} x$, where $\nu'\ge \nu$ and $x$ is a fixed string which does not start with $\rho$. 
Let $T[\ell..] = \rho^{\nu''} y$, where $\nu''$ is the maximal exponent possible.
If $\nu'' < \nu'-s$, then the order between $T[\ell..]$ and $T[r+sd..]$ coincides with the order between $x$ and $\rho$, if $\nu'' = \nu'-s$, then the order coincides with the order between $x$ and $y$, and in the case $\nu'' > \nu'-s$ the order is defined by the order between $\rho$ and $y$. 
It follows that to compute the number of suffixes $T[r_s..] \in \suf[i,j]$ bigger than $T[\ell..]$ we are to determine $\nu'$ and $\nu''$, and to compare at most three pairs of substrings of T. This can be done in constant time using the enhanced suffix array.
\end{proof}

\begin{theorem}
 	For any $\eps>0$ there is a data structure of size $\Oh(n)$, which can answer substring suffix selection queries in $\Oh(\log^{2+\eps}n)$ time.
\end{theorem}

\section{Conclusion}
In this paper we studied the substring suffix selection problem. We first revisited two special cases of the problem. For the problem of computing the minimal suffix of a substring we proposed a series of linear-space data structures with $\Oh(\tau)$ query time and $\Oh(n \log n / \tau)$ construction time, where $1 \le \tau \le \log n$ is an arbitrary parameter. We then showed that these data structures can be used to compute the Lyndon decomposition~\cite{chen1958free} of a substring of $T$ in $\Oh(k \tau)$ time, where $k$ is the number of distinct factors in the decomposition. For the maximal suffix problem we gave a linear-space data structure with constant query time and linear construction time. Both results improve upon the results of~\cite{Minmaxsuf}. Secondly, we studied the general case of the problem and showed that a string of length $n$ can be preprocessed into a linear-space data structure that allows to compute the $k$-th suffix of any substring of the string in $\Oh(\log^{2+\varepsilon} n)$ time.

Some problems remain open. First, we gave an optimal data structure for the maximal suffix problem, but not for the minimal suffix problem. Can constant query time, linear space and linear construction time be achieved? Secondly, the query time we gave for the general case of the problem is much worse than query times for the minimal or the maximal suffix problems. Can it be improved without changing the space bound?

\bibliographystyle{plain}
\bibliography{main}
\end{document}